\documentclass[submission,copyright,creativecommons]{eptcs}
\providecommand{\event}{TARK 2025} % Name of the event you are submitting to

\usepackage{iftex}
\usepackage{stmaryrd}
\usepackage{multicol}
\usepackage{comment}
\usepackage{booktabs}
\usepackage[english]{babel}
\ifpdf
  \usepackage{underscore,tabularx}         % Only needed if you use pdflatex.
  \usepackage[T1]{fontenc}        % Recommended with pdflatex
\else
  \usepackage{breakurl}           % Not needed if you use pdflatex only.
\fi

\usepackage{enumitem}

\usepackage{amsmath,amssymb,latexsym,amsthm}
\usepackage{xcolor}

\newcommand{\cD}{\mathcal{P}}

\newcommand{\cF}{\mathcal{F}}

\newcommand{\cL}{\mathcal{L}}
\newcommand{\cM}{\mathcal{M}}

\newcommand{\cP}{\mathcal{P}}

\newtheorem{definition}{Definition}
\newtheorem{example}{Example}
\newtheorem{proposition}{Proposition}
\newtheorem{lemma}{Lemma}
\newtheorem{corollary}{Corollary}
\newtheorem{theorem}{Theorem}

\newcommand{\LangInt}{\cL}   %full language we work with
\newcommand{\LangCPL}{\cL_{CPL}} %the language of CPL

\newcommand{\Prop}{\mathsf{Prop}} %the set of propositional variables
\newcommand{\I}{I}  %intention modality, in case we want to change the notation later
\newcommand{\imp}{\rightarrow}  
\newcommand{\biimp}{\leftrightarrow}
\newcommand{\Dia}{\Diamond}

\newcommand{\Log}{\mathsf{Log}}

\newcommand{\Nec}{\boxplus} %new necessitation operator

\newcommand{\s}{s}  %solution/topic assignment function
\newcommand{\D}{P}  %set of decision problems/topics
\newcommand{\PModel}{(P, \oplus, \s)} %problem model
\newcommand{\Frame}{(W, R, \cD, f)} %problem-senstive frame
\newcommand{\Model}{(W, R, \cD, f, V)} %  %problem-senstive model

\newcommand{\CModel}{\langle [\Gamma_0]_\Nec, R^c, \cD^c, f^c, V^c \rangle}  

\newcommand{\br}[1]{[\![ #1]\!]} %truthset

\newcommand{\daniil}[1]{{\color{red}{Daniil: #1}}}

\newcommand{\REF}[1]{{\color{red}{[REF]}}}

\newcommand{\Add}[1]{{\color{red}{Add: #1}}}

\title{Hyperintensional Intention}
\author{Daniil Khaitovich
\institute{ILLC, Univesity of Amsterdam\\ Amsterdam, The Netherlands}
\email{d.khaitovich@uva.nl}
\and
Ayb\"uke \"Ozg\"un
\institute{ILLC, Univesity of Amsterdam\\ Amsterdam, The Netherlands}
\email{a.ozgun@uva.nl}
}
\def\titlerunning{Hyperintensional Intention}
\def\authorrunning{Khaitovich \& \"Ozg\"un}

\begin{document}

\maketitle

\begin{abstract}

\textbf Intentions are crucial for our practical reasoning. The rational intention obeys some simple logical principles, such as agglomeration and consistency, among others, motivating the search for a proper logic of intention. However, such a logic should be weak enough not to force the closure under entailment; otherwise, we cannot distinguish between intended consequences of agents' choices and their unintended side-effects.  In this paper we argue that we should avoid not only the closure under entailment, but the weaker closure under equivalence as well. To achieve this, we develop a hyperintensional logic of intention, where what an agent intends is constrained by the agent's decision problem. The proposed system combines some elements of inquisitive and topic-sensitive theories of  intensional modals. Along the way, we also show that the existing closest relatives of our framework overgenerate validities by validating some instances of closure under equivalence. Finally, we provide a sound and strongly complete axiomatization for this logic.

%We show that the existing logics of intention overgenerate validities by forcing the latter closure principle. In order to solve the problem, we develop a hyperintensional logic of intention, which combines some elements of inquisitive and topic-sensitive theories of  intensional modals. We provide this logic with sound and strongly complete axiom system.
\end{abstract}
\section{Introduction}

%aybuke{to be revised further, depending on how the rest is structured and the contents of the sections}
Intentions do not obey many of the {\em closure principles} logicians bake in in their formalisms. While we are rationally committed to intend some of the consequences of what we intend, we are not rationally committed to intend, e.g., all logical, known, or believed consequences of our intentions. To cite a classic example \cite[p. 218]{cohen1990intention}, one intends to go to the dentists to get one's tooth filled, believing or  knowing that filling tooth will lead to being in pain (maybe because the agent doesn't know about anesthetics or they have a high tolerance to anesthetics or they cannot take much of them due to other health issues), while not intending to be in pain. Agents not only do not intend the unforeseen side-effects of their actions, but they also do not intend all of their foreseen side-effects \cite{bratman1987intention,cohen1990intention}.

This has already been well documented in the literature on logics of intentions \cite{cohen1990persistence,konolige1993representationalist,rao1997modeling,pollack1991overloading}, and much ink has been spilled to meet the challenge of finding a logical formalism for intentions that neither over- nor undergenerates the {\em intended} consequences of one's intentions  such that the overgeneration problem is solved in a principled, well-motivated way \cite{chen1999logic,beddor2023question,roy2008thinking,van1996formalising}. This is exactly what this paper aims to do. It provides a semantics for intentions that at least alleviates, if not completely avoids, the overgeneration problem - also known as {\em the problem of side-effects} \cite{Finnis_1991} - for the right reasons and that explains which closure principles should be preserved.

The key idea we fall back on to achieve our goal is the so-called {\em issue-relevance} \cite{berto2018aboutness,Kroon2024}  or {\em question-sensitivity} \cite{minica2011dynamic,yalcin2018belief,hoek2022questions,baltag2018group} of propositional mental states, which has been explored extensively in a body of rapidly growing literature. 
%Most of such work models epistemic attitudes such as knowledge and belief and aims to solve the infamously difficult problem of logical omniscience that is built into many epistemic logics \cite{stalnaker1991problem,vardi1986epistemic,hawke2020fundamental} or to account for the effects of inquiry on a rational, idealize agent's knowledge and belief \cite{baltag2018group,Kroon2024}. 
Question-sensitivity of, e.g.,  epistemic/doxastic attitudes can explain the failure of some closure principles by modeling these attitudes as dependent on inquiry relevant to an agent's epistemic/doxastic agenda, alleviating the infamous problem of logical omniscience \cite{stalnaker1991problem,vardi1986epistemic,hawke2020fundamental} or accounting for the effects of inquiry on a rational, idealized agent's knowledge and belief \cite{baltag2018group,Kroon2024}. In a similar vein, observing the parallel between the problems of logical omniscience and side-effects, recent work \cite{beddor2023question} developed a question-sensitive theory of rational intentions, modeling intentions as dependent on the agent's practical question {\em what to do?}, namely, their decision problem. We take our cue from this theory but argue that it still overgenerates validities: it is still subject to the problem of side-effects. To briefly explain (and to be further elaborated below), the purely possible-worlds, partition-based mechanism employed in \cite{beddor2023question} avoids closure under logical, necessary or known/believed entailments but still validates closure under (undesired instances of) logical and necessary {\em equivalents} for intentions. We provide compelling examples that challenge closure under logical, necessary, and known/believed equivalents and, in turn, argue that a logic of intentions should also avoid these closure principles of equivalents. In particular, it should be {\em hyperintensional}.

%\footnote{\cite[Section 6.2]{beddor2023question} presents a semantics that models questions as subject matters and avoids full closure under logical and necessary equivalents. This formalism still validates these closure principles for atomic propositions. For us, even such restricted versions of  closure under logical and necessary equivalents are too strong, leading to the problem of side-effects: see Examples \ref{exmpl:failure_of_cl_under_log_equiv} and \ref{exmpl:pacifist_taxpayer}, and Section \ref{sec:partitions} for a detailed explanation.} 
Our proposal employs tools from topic- or subject matter-sensitive semantics for knowledge and belief \cite{ozgun2021dynamic,berto2022topics,hawke2020fundamental} and use them to model {\em decision problems} in a finer, hyperintensional way.  We fully develop this formalism based on a bi-modal language with a global and an intention modality, and provide a sound and strongly complete axiomatization for the proposed hyperintensional logic of intention.

The paper is organized as follows. In Section 2, we present the problem of side-effects and challenges it poses for any logic of intention: we argue that closure under (logical, necessary and known/believed) equivalence is too strong for the logic of intention and list some principles that the logic of intention should validate. In Section 3, we review the closest relatives of our framework, namely the proposals in \cite{beddor2023question}, and show that they do not satisfy the requirements we state in Section 2. In Section 4, we present a hyperintensional logic of intention that is weak enough not to force the closure under equivalence, but strong enough to validate all ``good'' principles. The framework is closely related to topic-sensitive theories of doxastic propositional attitudes \cite{ozgun2021dynamic,berto2018aboutness}.  Moreover, we provide a sound and strongly complete axiom system for the hyperintensional logic of intention. As the proofs are not central to the conceptual contributions of this work, the longer proofs are omitted from the main body of the paper and presented in the appendices.

\section{Logics of Intension and The Problem of Side-Effects}\label{sec:side-affects}

Formal epistemologists and philosophical logicians often turn to Kripke semantics to model propositional mental attitudes such as knowledge, belief and, of particular interest for us in this paper, {\em intention}. As well known, however, this mainstream approach validates certain closure principles for these attitudes that are fit to model only highly idealized knowers and believers. As also emphasized, e.g., in \cite{konolige1993representationalist}, these closure principles are a more serious threat for intention, as they cannot be defended for intentions even as idealizations. In this paper, we will focus on the following closure principles, formulated for intentions via operator $
I$:
%aybuke{Check for exectly which ones we need, given their strengths and what we have later in the paper.}
\begin{enumerate}
\item\label{CLE} {\bf Closure under logical entailment:} $ \models \varphi\rightarrow \psi$ then $ \models \I\varphi \rightarrow \I\psi$
\item\label{AE} {\bf Necessary entailment:} $\models \Nec(\varphi \rightarrow \psi)\rightarrow (\I\varphi \rightarrow  \I\psi)$
%\item\label{CLKE} {\bf Closure under known implications} $\models K(\varphi \rightarrow \psi)\rightarrow (\I\varphi \rightarrow  \I\psi)$
\item\label{CLEq} {\bf Closure under logical equivalence:} $\models \varphi \leftrightarrow \psi$ then $\I\varphi \leftrightarrow  \I\psi$
\item\label{AEq}  {\bf Necessary equivalence:} $\models \Nec(\varphi \leftrightarrow \psi)\rightarrow (\I\varphi \leftrightarrow  \I\psi)$
%\item\label{CLKEq} {\bf Closure under known equivalences} $\models K(\varphi \leftrightarrow \psi)\rightarrow (\I\varphi \leftrightarrow  \I\psi)$
%item\label{KAEq}  {\bf Closure under known a priori equivalence} $\models K\Nec(\varphi \leftrightarrow \psi)\rightarrow (\I\varphi \leftrightarrow  \I\psi)$
\end{enumerate}

A few notes on notation is warranted. In the formulation of principles \ref{CLE}-\ref{AEq}, $\models$ represents the logical consequence relation of the relevant logic, $\I\varphi$ is read as ``the agent intends that $\varphi$''. $\boxplus\varphi$ is read as ``it is a priori that $
\varphi$'' or ``the agent knows that $\varphi$'', depending on the context and examples we target. The reason why we do not introduce distinct operators for knowledge and a priori necessity is merely practical and the lack of such distinct operators does not bear on any substantive conceptual points we want to make in this work.  In our formalism, we interpret $\boxplus$ as the global modality, closer to the standard interpretation of a priori necessity. Introducing an additional knowledge operator $K$ interpreted on a subset of the logical space, representing epistemically accessible worlds of the agent (as in, e.g., the Kripke semantics of epistemic logic) does not make any difference for the kinds of logical principles we focus on in this work.\footnote{This claim assumes that the set of so-called conative alternatives w.r.t. which the operator $\I$ is interpreted is a subset of the epistemically accessible worlds. Moreover, one can also read $\boxplus$ as belief. In this extended abstract, we often refer to knowledge to keep the presentation more concise.} This much formalism should be sufficient for our exposition in Sections \ref{sec:side-affects} and \ref{sec:partitions}. All these notions and notation will be properly introduced in Section \ref{sec:logic}.

 We say that  a logic {\em suffers from the problem of side-effects} if it validates at least one of \ref{CLE}-\ref{AEq} (even when $\varphi$ and $\psi$ are propositional variables). Just as an epistemic agent can be non-omniscient for diverse reasons \cite{berto2022topics,hawke2020fundamental}, the problem of side-effects can be explained by different factors. Since the early stages of intention logic research, principles \ref{CLE} and \ref{AE} have been considered undesirably strong  \cite{konolige1993representationalist,chen1999logic}. There are at least two factors explaining the failure of these principles:  
\begin{enumerate}
    \item  \textit{Epistemic flaws}: One can intend that $\varphi$ only if one has the epistemic resources to reason about $\varphi$. Agents cannot derive and do not always know/believe all the consequences of what they intend, thus, they do not always intend all necessary or logical consequences of what they intend. 
\end{enumerate}    
    
\noindent For example, one may intend to maintain a vegan diet without intending to avoid eating anything that contains cysteine, simply because one has no idea what cysteine is or does not know that it is made out of animal products.\footnote{This example is inspired by Stalnaker's famous William III example for belief \cite{stalnaker1984inquiry}, which is also used in \cite[p. 348] {beddor2023question} for intention to make a similar point.}
\begin{enumerate}[resume]
    \item \textit{Control constraint}:  One can intend that $\varphi$ only if one believes that one has control over $\varphi$. So that, if one intends  that $\varphi$, $\varphi$  entails $\psi$, but  one does not believe that one has control over $\psi$, one will not intend $\psi$.   \cite[p.357]{beddor2023question}.
\end{enumerate}
%\aybuke{I think there is also another argument to make: Epistemic Flaws and Control do not only motivate the failure of closure under logical entailment, but also motivate the failure of closure under logical equivalence. This is missing in Beddor \& Goldstein (and maybe others). }
%\daniil{I have reformulated Introduction to pose exactly that argument. The grounding thing is now simply incorporated in the definition of control constraint.}

%MAYBE ADD SOMEWHERE APPROPRIATE
%Just as an agent can be non-omniscient for diverse (but primarily epistemic) reasons \REF, the problem of side-effects can be explained by different factors.

%As exemplified by the tooth-pulling example in the introduction, not all consequences of ones intentions are their intentions, even those they fully foresee coming. This suggests that there is more at play than an agent's epistemic flaws in determining what intentions follow from what one intends.

Epistemic flaws are not specific to intention and are widely studied in formal theories of knowledge and belief (see, e.g., \cite{Fagin2003,vardi1986epistemic,stalnaker1984inquiry,hawke2020fundamental}). On the contrary, control constraint is a unique feature of intention and similar motivational attitudes.
%It poses a challenge not only to principles \ref{CLEq} and \ref{AEq} involving equivalences (rather than entailments), but also to \ref{CLKE}-\ref{KAEq}, involving {\em knowledge} of the relevant entailments and equivalences. 
Cases of the failure of closure under entailment due to control constraint  have been discussed in the existing literature \cite{roy2008thinking,beddor2023question}. Authors often point to \textit{foreseeing without intending}: one may intend to do some action, foresee some necessary consequences of that action (hence know that the intended action entails such consequences), yet still not intend  these consequences \cite{aune1966intention}. We focus on the less discussed  cases, where even the weaker principle of closure under equivalence fails.\footnote{In \cite{chen1999logic}, authors approach the issue from the perspective of resource-bounded agents who can only conceive the so-called {\em cognitively finite objects}. Instead of possible worlds, which are maximally consistent sets of formulae that might lack finite cognitive representations,  they use \textit{cognitive abstractions}, i.e. consistent (but not necessarily maximally consistent) set of formulae that are  built from finitely many propositional variables. This solution allows one to avoid the closure under logical equivalence, but is not sufficient to explain cases of foreseeing without intending: the agent may have enough cognitive resources to grasp and foresee some consequence of their intention, but might still not intend it. A similar worry can also be raised against  \cite{van1996formalising}, where all undesired closure principles are avoided via the machinery of awareness functions. This is conceptually insufficient, since, as foreseeing without intending demonstrates,  agents might be aware of some unintended consequences of their intended actions. A similar solution based solely on syntactic restriction is given in \cite{van2007towards}: closure under logical equivalence (and entailment) is restricted by  arbitrary \textit{practical rules of reasoning}, by means of which we can avoid not only undesired instances of such closure principles, but also virtually any closure, such as agglomeration.}

Below we present a few examples motivating the failure of \ref{CLEq}-\ref{AEq} exactly due to control constraint. In all such cases, the bearer of intention intends that $\varphi$, knows that $\varphi$ is (necessarily or materially) equivalent to $\psi$, but does not intend that $\psi$. Since the equivalences are known by the agent in question, such failures cannot be explained by epistemic flaws.

\begin{example}
Imagine a patient who suffers from a severe dental disease. The doctor tells the patient that without any treatment, he will certainly lose his teeth, but the treatment decreases the chance only by 50\%. By agreeing to the treatment, the patient intends to save his teeth with 50$\%$ chance, but he doesn't intend to lose his teeth with 50$\%$ chance.\footnote{Similar examples can be found in the literature on framing effects \cite{kahneman1984choices}.}
\label{exmpl:failure_of_cl_under_log_equiv}
\end{example}

%aybuke{We need to shorten examples 2 \& 3, and maybe change the military part of exmaple 2. Due to space restrictions, we might need to eliminate some.}
\begin{example}
In some monarchy, it is obligatory to pay taxes. A citizen of the monarchy does not want to face any legal consequences of tax avoidance, hence, she  intends to pay taxes. The  monarch necessarily spends tax money on official ceremonies. In this context, the  propositions ``\textit{one's taxes are paid}" and ``\textit{one's paid taxes are spent on official ceremonies}" are  necessarily equivalent (in all the worlds that comply with the customs of the monarchy): if one pays taxes, then the monarch spends them on ceremonies, while the spending one's paid taxes on ceremonies entails that one's taxes are paid. Being familiar with the customs of the monarch, this equivalence is known by the citizen. Nevertheless, the citizen may not intend that her  paid taxes  are spent on official ceremonies, despite intending to pay taxes.
\label{exmpl:pacifist_taxpayer} 
\end{example}

\begin{comment}
\begin{example}
A lecturer teaches courses in computer science  and logic. One day she gives a lecture in  computer science class, the other day she gives a lecture in logic class. The only lecture that she gives for both classes is on Curry-Howard correspondence. 

After she explained Curry-Howard correspondence to students of both classes, she knows that all students understand it well enough to apply it to simple examples. Then, during the logic class, she receives a question from a student who attend both classes: ``\textit{How does material implication behave}?" She answers the question and intends that the student will  have enough information to understand how the material implication behave. He has enough information to understand  the material implication  iff he has enough information to understand function type: using the knowledge about Curry-Howard correspondence, he can explain one it terms of another. The teacher knows about this ability of this student, nevertheless, it is natural to assume that teacher had no intention to explain behaviour of function type to the student while answering the question about material implication.
\label{exmpl:lecturer}
\end{example}
\end{comment}

%aybuke{Let's very carefully check the uses of "known equivalences", "knonw a priori equivalences" etc. We do not want to make a simple mistake here.}

%aybuke{Also add somewhere earlier that entailment rules have been discussed in the literature quire a bit. Thus we focus on equiavelences.}

\noindent We cannot explain these failures of the closure under equivalence by epistemic flaws. Instead, they can  be explained via control constraint.\footnote{While early theories of agency  define control strictly in terms of causation \cite{chisholm2018human} -- the agent has control over $\varphi$ iff they can force $\varphi$ by their actions -- the more recent accounts refine these definitions, connecting control with explanatory relations.  In \cite{KelleyForthcoming-KELACT-3}, it is argued that an agent has control over $\varphi$ only if the agent can act in a way that will lead to $\varphi$ and the explanation why $\varphi$ will obtain is centered around the actions of the agent. This leads to the failure of the closure under equivalents: it is possible that $\varphi$ holds iff  so does $\psi$, but explanations why those propositions are true differ. Explanation is often seen  as an example of hyperintensional phenomena. All mathematical truths are necessary true and hence logically and a priori equivalent, nevertheless, they have different explanations. Such phenomena occur with contingent propositions as well: ``\textit{it is true that grass is green}'' is explained by ``\textit{the grass is green}'', but not vise versa, while two propositions are obviously necessarily equivalent. Since the metaphysics and logic of explanation are out of the scope of the present paper, we do not elaborate further on this argument. See, e.g., \cite[p.157]{nolan2014hyperintensional} and \cite{schnieder2011logic} for further discussion.}  Intuitively, the patient has control over the decrease of the chance of teeth loss by 50\%, since there is an option to take the treatment, but has no control over the fact that there still will be 50\%  chance of losing the teeth. The taxpayer controls whether her taxes are paid or not, but it is not in her control to prevent the monarch from spending her tax money not to her liking.\footnote{As pointed out by one of the reviewers, the last case resembles familiar examples of intentions that involve the actions of others \cite[Chapter 8]{bratman1999faces}. Since our current framework is single-agent, what matters for the purposes of this paper is whether the agent has control over the relevant proposition -- regardless of whether the lack of control stems from another agent’s actions or from something else. Exploring how a lack of control arises from the actions of others in a multi-agent extension of this framework is an interesting direction for future work, which we plan to pursue in a follow-up paper.}

Therefore, the control constraint poses a challenge not only to the principles of closure under entailment (\ref{CLE} \& \ref{AE}), but also to the principles of closure under equivalences (\ref{CLEq} \& \ref{AEq}). A logic of rational intention should avoid them. Yet, there are a number of simple and intuitive logical principles that rational intention should validate; that is, a logic of intention is still possible. For example, the restricted versions of the above mentioned closure principles, constrained to the cases where the agent has control over both propositions in question and there are no epistemic obstacles, should be validated. Moreover, it is broadly endorsed that rational intentions are consistent (if one intends that $\varphi$, one doesn't intend that $\neg \varphi$) \cite{cohen1990intention,konolige1993representationalist,chen1999logic,roy2008thinking}  and they agglomerate (if one intends that $\varphi$ and one intends that $\psi$, then one intends that $\varphi \land \psi$) \cite{bratman1987intention,beddor2023question}.\footnote{Admittedly, the agglomeration principle for intention is more contentious. \cite{sverdlik1996consistency}, for example, argues that the unqualified agglomeration principle pressures agents to have ``{\em one enormous compound intention}'' [p. 517] that guide their actions, which is arguably implausible and might hinder their ability to satisfy their relevant desires. Still, there are ways to defend the agglomeration of intentions against such objections: see, e.g.,  \cite{zhu2010principle} for an elaborate discussion. In this paper we adopt the agglomeration principle and leave a more detailed analysis of it for future work.} %In the following section, we address the existing formal solutions to that problem and show that they overgenerate validities via the closure under equivalences.

\begin{comment}
    Before moving forward with our formalism, we explicitly state the requirements our examples and philosophical intuitions pose to any formal theory of intention:
\begin{itemize}
    \item Agent's intention should not be closed under known a priori entailment and equivalence;
    \item Agent's intention should be consistent: if one intends that $\varphi$, one cannot rationally intend that $\neg \varphi$ at the same moment;
    \item Agent's intention should agglomerate: if one intends that $\varphi$ and one intends that $\psi$, then one is rationally required to intend that $\varphi \land \psi$;
    \item Agent's intention should be closed under purely syntactic reformulations. For example, if one intends that $\varphi \land \psi$, then one is rationally pressured to intend that $\psi \land \varphi$, that $\varphi$ and that $\psi$, that $(\varphi \land \psi) \lor \psi$, etc.
\end{itemize}

 %aybuke{****I have partially revised up to here and the axiomatization section! Proofs are moved to the appendix****}
\end{comment}

\section{Formalizing Decision Problems as Partitions}\label{sec:partitions}

%aybuke{Here discuss the two extremes: partitions and abstract objects. Argue that we need something in between. }

One of the recent solutions to the  problem of side-effects  borrows another tool from formal epistemology -- question-sensitivity. Some have argued that our knowledge and beliefs exist in the context of \textit{questions} \cite{stalnaker1984inquiry,yalcin2018belief,baltag2018group,Kroon2024}. Likewise, one can argue that our intentions exist in the context of \textit{decision problems}. Just like any belief we hold is an answer to some question, any intention we have is a solution to a problem ``\textit{what to do?"}  \cite{beddor2023question}. Question-sensivity explains why  we may not intend  some consequences of our intentions, even when we know what follows from them: some propositions  might be \textit{undefined} in the context of our decision problems; they might not constitute (partial) solutions to them.  

%\subsection{Decision problems as partitions}
How to formally model a decision problem? So far, authors take a decision problem to be a question and formalize it in terms of a partition of a logical space  \cite{beddor2023question,hoek2022questions,baltag2018group,van2012toward,minica2011dynamic}:

%aybuke{Maybe it is better to denote a decision problem by $\Pi$ or $D$, for ``partition'' or ``decision'', respectively.}
\begin{definition}[Decision problem as partition]
Given a non-empty set of possible worlds  $W$, a decision problem $\Pi$ is a partition of $W$.\footnote{A partition $\Pi$ of $W$ is a subset $\Pi\subseteq 2^W$ such that $\emptyset\not \in \Pi$, $\bigcup\Pi=W$, and for any $X, Y\in \Pi$ such that $X\not = Y$, $X\cap Y=\emptyset$.} 
\label{def:questions_as_partitions}    
\end{definition}
\begin{definition}[Complete and partial solutions] 
Every cell $X \in \Pi$ is a complete solution to the decision problem $\Pi$. A union of any set of cells $A = \bigcup_{i \in I} X_i$ for some  $\{X_i\}_{i \in I} \subseteq \Pi$ is a partial solution to $\Pi$.  
 \label{def:solutions_resolutions}
\end{definition}

It is easy to see that every complete solution $X \in \Pi$ is a partial solution to $\Pi$: $X=\bigcup\{X\}$ with $\{X\}\subseteq \Pi$. % and absurdity is a partial solution to any problem, since $\emptyset = \bigcup \emptyset \subseteq \Pi$ for any  problem $\Pi$.

Intuitively, we associate a decision problem with a set of its mutually exhaustive complete solutions that the agent considers implementing. %We call such solutions \textit{total} and model them as equivalence classes under the corresponding partition. 
\textit{Partial} solutions may be viewed as indeterministic choices between a set of compelete solutions: agents may intend only to partially solve their decision problem, given that sometimes the difference between some complete solutions is of no importance to them. For example, one may be dealing with the quest of getting to the railway station and consider three possible ways to do that: by bus, by tram or by foot. Both the bus and the tram depart from the same place and take roughly equal amount of time to get to the destination, so one may intend to simply go to the bus/tram station and take a bus or a tram without settling on either of the two complete solutions. 

%aybuke{The subject matters come all of a sudden, out of nowhere? Why are they relevant? The reader is lost at this point. Also, notation is not explained. What are the semantics? Add a note saying that we work with the language of classical propositional logic and assume its standard semantics.}
A decision problem restricts what propositions the agent can intend.  Namely,  if one intends that $\varphi$, then $\varphi$ should be {\em defined on the agent's decision problem}. Beddor \& Goldstein \cite{beddor2023question} has two proposals to formalize the latter notion of definedness on a decision problem. To recap briefly, the first one takes it that $\varphi$ is defined on the agent's decision problem $\Pi$ iff $\varphi$ is a partial solution to $\Pi$ (in the sense described in Definition \ref{def:solutions_resolutions}). The corresponding logic of intention invalidates closure under entailment and validates agglomeration, but still forces closure under equivalents in full generality. The second proposal is a closer rival to our account presented in this paper, so we introduce its formal components in more detail and argue that it still overgenerates validities. Below, we assume that the object of intention is a sentence of a language of classical propositional logic, i.e., a formula that is built from atomic propositions using classical Boolean connectives ($\neg, \land, \lor$.  etc). The semantics for such expressions is the standard possible worlds semantics, where $\llbracket \varphi \rrbracket$ denotes the set of possible worlds that make $\varphi$ true.

The second proposal (see \cite[Section 6]{beddor2023question}) takes it that a proposition $\varphi$ is defined on a decision problem iff the \textit{subject matter} of $\varphi$ is included in the decision problem, where the notion of subject matter is also defined via partitions, in line with Lewisian theory of subject matters \cite{lewis1988relevant,lewis1988statements}.\footnote{A detailed presentation of theories of subject matter is beyond the scope of this work. We refer the reader to \cite{lewis1988statements,hawke2018theories,Plebani2021,lewis1988relevant}.}  Here, we briefly restate the relevant definitions of \cite{beddor2023question}.%\footnote{The first proposal in \cite{beddor2023question} states definedness on a decision problem as follows: $\llbracket \varphi \rrbracket$ is defined on the agent's decision problem $\Pi$ iff $\llbracket \varphi \rrbracket$ is a partial solution to $\Pi$, that is, $ \llbracket \varphi \rrbracket= \bigcup_{i \in I} X_i$ for some  $\{X_i\}_{i \in I} \subseteq \Pi$. It is easy to see that, if we replace this by (2) in Definition \ref{defn:qsensitive-intention}, the resulting logic satisfies \ref{CLEq} and \ref{AEq}.\label{footnote:proposalone}} 
\begin{definition}[Subject matters]
Given a logical space $W$ and any formula $\varphi$,  $sm(\varphi) \subseteq 2^{2^W}$ is the subject matter of $\varphi$, defined recursively as follows (where $p$ is atomic):
    \begin{center}
\begin{tabular}{lll} 
$sm(\top) = \{W\} $ & \,\, & $sm(\neg\varphi) = sm(\varphi)$\\
$sm(p) = \{ \llbracket p \rrbracket, W \setminus \llbracket p \rrbracket\} \setminus \{\emptyset\}$ &\,\, & $sm(\varphi \land \psi) = sm(\varphi \lor \psi) = \{ 
 X \cap Y \: |\: X \in sm(\varphi), Y \in sm(\psi)\}\setminus \{\emptyset\}$
\end{tabular}
\end{center}

\begin{comment}
\begin{itemize}
    \item $sm(\top) = sm(\bot) = \{W, \emptyset\}$
    \item For any atomic proposition $p$: $sm(p) = \{ \llbracket p \rrbracket, W \setminus \llbracket p \rrbracket\}$;
    \item $sm(\neg\varphi) = sm(\varphi)$;
    \item  $sm(\varphi \land \psi) = sm(\varphi \lor \psi) = \{ 
 X \cap Y \: |\: X \in sm(\varphi), Y \in sm(\psi)\}$
\end{itemize}
\end{comment}

Note that  for any formula $\varphi$, $sm(\varphi)$ is a partition on $W$\footnote{Our definition of $sm$ function differs from the one given in \cite{beddor2023question}. The only differences are that we define $sm(\top)$ as a primitive notion and make sure that for any $\varphi \in \cL$, $\emptyset\not\in sm(\varphi)$. It affects the framework neither conceptually nor technically; and is done for the sake of consistency with our notation and definitions.}; Moreover,  $\exists P \subseteq sm(\varphi): \bigcup P =  \llbracket \varphi \rrbracket$ and $ \exists N \subseteq sm(\varphi): \bigcup N = \llbracket \neg \varphi \rrbracket$. 
\label{def:subject_matters}
\end{definition}
\begin{definition}[Parthood on subject matter]
Given a logical space $W$ and two propositions $\varphi, \psi$, a subject matter of $\varphi$ is  a part of subject matter of $\psi$, $sm(\varphi) \sqsubseteq sm(\psi)$, iff for any $X \in sm(\varphi)$ there exists a subset $S \subseteq sm(\psi)$, such that $\bigcup S = X$.
\label{def:parthood_on_sm}
\end{definition}

\begin{comment}
Subject matters are modeled as partitions by the following reason. Given a subject matter $sm(\varphi)$ and a partition cell $X \in sm(\varphi)$, all world in $X$ are exactly alike w.r.t. the subject matter of $\varphi$:
\begin{quote}
    \textit{We can say that two worlds are exactly alike with respect to a given subject matter. For instance two worlds are alike with respect to the 17th Century iff their 17th Centuries are exact intrinsic duplicates (or if neither one has a 17th Century)}.\cite[p.161]{lewis1988relevant}
\end{quote}
For a detailed explanation of this view of subject matters, we refer to \cite{lewis1988relevant,lewis1988statements}.
\end{comment}

We are now set up to define the formal framework. Given some logical space $W$, we can represent  intention via (1) the set of conative alternatives, i.e. the non-empty set of possible worlds $Con \subseteq W$, such that the agent's intention is satisfied in $Con$-worlds; (2) the decision problem  $\Pi$ , which is a partition of the logical space $W$.  

%aybuke{Decision problems are denoted by $D$ from now on. Let's fix the notation before.}

\begin{definition}[Question-sensitive intention]
Given a logical space $W$, a set of conative alternatives $Con$ and a decision problem $\Pi$, an agent intends that $\varphi$ iff (1) $Con \subseteq \llbracket \varphi \rrbracket$ and (2) $sm(\varphi) \sqsubseteq \Pi$, i.e. for any $X \in sm(\varphi)$ there exists  a partial solution $S \subseteq \Pi$, s.t. $X = \bigcup S$.
\label{def:definability_revised}
\end{definition}

The intuition behind Definition \ref{def:definability_revised} is the following: the agent intends that $\varphi$ iff (1) they intend to solve their decision problem in a way that will force $\varphi$ and (2) $\varphi$ is defined on the agent's decision problem. The agent has control over $\varphi$, since $\varphi$ing is considered a partial solution to their decision problem. By the very same reason, the agent is aware that they bring about that $\varphi$ by the solution they are committed to implement. Hence, both control constraint and epistemic flaws are accounted for in this semantics.

Definition \ref{def:definability_revised} validates consistency and agglomeration for  intention. The closure under a priori entailment (\ref{AE}) and equivalence (\ref{AEq}) fail, since the scope of the intention is restricted by the decision problem.  Nevertheless,  the usage of subject matters does not rule out all cases when the closure under  a priori or known equivalences may fail. The closure cannot be falsified for atomic propositions: it not hard to see that under Definition \ref{def:definability_revised}, for any atomic propositions $p, q$, if $\llbracket p \rrbracket = \llbracket q \rrbracket$, then $sm(p) = sm(q)$ and hence one intends that $p$ iff one intends that $q$. Examples \ref{exmpl:failure_of_cl_under_log_equiv} and \ref {exmpl:pacifist_taxpayer}  falsify this principle. 

Yet another, more conceptual, argument against formalizing decision problems as partitions comes from the way the notion of parthood on questions is defined. Intuitively, one decision problem is a part of another iff every complete solution to the former is a partial solution to the latter (a complete solution to a coarser problem is a partial solution to a finer problem). More formally, $\Pi_1$ is a part of $\Pi_2$, denoted by $\Pi_1 \sqsubseteq \Pi_2$,  iff for all $A\in \Pi_1$ there is $\{X_i\}_{i \in I} \subseteq \Pi_2$ such that $A=\bigcup_{i \in I} X_i$. We refer to $\sqsubseteq$ relation as \textit{extensional parthood}, i.e. $\Pi_1 \sqsubseteq \Pi_2$ reads as ``\textit{$\Pi_1$} is an extensional part of $\Pi_2$''. Another useful acronym: $\Pi_1 \simeq \Pi_2 : = \Pi_1 \sqsubseteq \Pi_2  \: \& \: \Pi_2 \sqsubseteq \Pi_1$, referred to as \textit{extensional equivalence} ($\Pi_1$ is extensionally equivalent to $\Pi_2$). The claim according to which there is nothing more to parthood than an extensional parthood, is going to be referred as \textit{extensional parthood claim}. Modeling decision problems as partitions supports the extensional parthood claim, affecting both proposals of question-sensitivity in \cite{beddor2023question}.  On the contrary, we reject the extensional parthood claim, since there are decision problems $\Pi_1, \Pi_2$, where $\Pi_1\sqsubseteq \Pi_2$, but the agent can intentionally solve $\Pi_1$ without any intention to solve $\Pi_2$. Consider the next example.
\begin{example}[Taking a train service without minding the time and place of departure]
A ticket office clerk works at a train station with two platforms. At every platform, a train departs once an  hour, at $j$:00. Every train has a unique code of the form $(i;j)$:  the train $(i;j)$ departs from the $i$th platform at $j$:00. Let $r^i_j$ be  the proposition ``\textit{someone has a ticket to the train with service number $(i;j)$}", $t_j$ -- ``\textit{someone has a ticket to a train that departs at $j$:00}", $p_i$ -- ``\textit{someone has a ticket to a train that departs from the $i$th platform}". Note that every atomic proposition $r^i_j$ is equivalent to the corresponding conjunction $p_i \land t_j$: every train service has its unique departure platform and time. 

Someone comes to the train station and asks the clerk to sell them a ticket for a train with number $(2;13)$: the customer only knows the train service number, since they were asked to buy that ticket with no further details provided. The clerk knows the system well and is aware that the train with number $(2;13)$ departs from platform 2 at 13:00. There is a pile of 48 tickets with numbers: $(1;1), (1;2), \ldots, (2;24)$. The clerk needs to choose which ticket they should take from the pile and intends that the customer gets the ticket for the train with service number $(2;13)$. The clerk intends neither that the customer gets a ticket for a train that departs from the second platform nor that the customer gets a ticket for a train that departs at 13:00, since that is simply not the focus of the clerk.
\label{exmpl:focused_clerk}
\end{example}

Intuitively, Example \ref{exmpl:focused_clerk} is just another instance of foreseeing without intending.
The clerk intends to give the customer the ticket the customer has asked for. The clerk also foresees that as a consequence of that, the customer will receive a ticket for a train that departs at 13:00 from platform number 2. Nevertheless, the clerk does not intend to give the customer a ticket for a train with a specific departure time or place, since it is just irrelevant to the decision problem the clerk is facing. Unfortunately, formalization via subject matters of solutions cannot capture this nuance. The clerk needs to choose the ticket from the pile, i.e., they consider all the tickets, and their decision problem is $\Pi = \{ \llbracket r^i_j \rrbracket \: | \: i\in \{1,2\}, 0 \leq j \leq 24\} \cup \{W \setminus \bigcup \limits^{i \in \{1,2\}}_{0 \leq j \leq 24} \llbracket r^i_j\rrbracket\}$. In other words, $\Pi = sm(\bigvee \limits^{i \in \{1,2\}}_{0 \leq j \leq 24} r^i_j)$.  Under Definition \ref{def:parthood_on_sm}, $sm(\bigvee \limits^{i \in \{1,2\}}_{0 \leq j \leq 24} r^i_j) = sm(\bigvee \limits^{i \in \{1,2\}}_{0 \leq j \leq 24}(p_i \land t_j))$. So that, $r^i_j$ is defined on $\Pi$ iff $p_i \land t_j$ is defined on  $\Pi$. Under Definition \ref{def:definability_revised}, the clerk cannot intend that $r^i_j$ without intending that $p_i$ and that $t_j$.

We may come up with a number of similar examples: one may look for a specific building without minding the street and the house number of that building or wonder about one's age  without thinking about neither the day nor the month nor the year one was born, etc.: in all such cases, an agent solves a problem without minding the extensional parts of the problem. 

Summarizing, modeling decision problems as partitions does not seem to be sufficiently fine-grained, since it enforces extensional parthood and, as a consequence, the closure under equivalence \ref{AEq} is satisfied in contexts where it should fail.

\section{A Hyperintensional Logic of Intention}\label{sec:logic}
As we have just seen, modeling decision problems as partitions overgenerates validities: it validates undesired, atomic instaces of $\ref{AEq}$. Luckily, there is a concurrent way to represent decision problems: as atomic objects. Namely, we can simply state that there is a set $P$ of problems that do not have internal structure. Then, we can define a solution function $s$, which maps every formula $\varphi$ to a set of problems: if $a \in s(\varphi)$, then bringing about that $\varphi$ is a partial solution to $a$. Definedness on a decision problem and control constraint are expressible via $s$. 
If the agent is solving a decision problem $a \in P$ and $a \in s(\varphi)$, then bringing about that $\varphi$ is a partial solution to $a$ the agent considers. Hence,  $\varphi$ is defined on $a$, which means that the agent has control over $\varphi$ and is aware that some solution to the decision problem forces $\varphi$.  To have an adequate notion of parthood on decision problems, we  define the corresponding partial order $\leq$ on $P$. Then, following our intuition about parthood,  function $s$ should be $\leq$-upward closed: if $a \leq b$, then for any $\varphi$,  $a \in s(\varphi)$ entails $b \in s(\varphi)$. Informally, if $a$ is a part of $b$, then any partial solution to $a$ is a partial solution to $b$. As we are to see, this approach steps on the hyperintensional terrain, where the closure under  logical and necessary equivalences fail.  Moreover, it will allow us to differentiate even between the necessarily equivalent atomic propositions and adequately model cases such as Examples \ref{exmpl:failure_of_cl_under_log_equiv}-\ref{exmpl:focused_clerk}. We also elaborate on how our proposed way of modeling decision problems relates to the partition-based perspective presented in Section \ref{sec:partitions}.

\subsection{Formalizing decision problems as atomic objects}

First, we define the formal language we are to work with. Let $\LangCPL$ be the language of classical propositional logic with a countable set of atomic propositions $\Prop$, defined recursively as follows: 

$$
\alpha : = \top \: | \: p \: | \: \neg \alpha \: | \: (\alpha \lor \alpha) \: | \: (\alpha \land \alpha) 
$$
where $p\in \Prop$. Note that we have  $\top$, $\neg$, $\lor$ and $\land$ as primitives and do not define one in terms of another. Nevertheless, we use the interdefinability of $\neg$, $\lor$ and $\land$ in terms of each other as our logic cannot discern, e.g. $\varphi \land \psi$ from $\neg(\neg \varphi \lor\neg \psi)$. However, we have specifically defined $\top$ as a primitive: it is pure tautology without non-logical context. The reader should not confuse $\top$ with another tautology expressed using the elements of $\Prop$, such as $p\lor \neg p$ and $p\rightarrow p$. They will be logically equivalent but they won't be substitutable with each other under the scope of our intention modality. For $\bot$, we set $\bot:=\neg \top$. Based on $\LangCPL$, the well-formed formulas of the language $\LangInt$ of our logic of intention are given by the following grammar: 
$$
 \alpha \: | \: \neg \varphi \: | \: (\varphi \lor \varphi) \: |\: (\varphi \land \varphi) \: | \: \Nec \varphi \: | \: \I \alpha
$$
where $\alpha\in\LangCPL$.  $\I \alpha$ reads as ``\textit{agent intends that $\alpha$}''.  Notice that the maximum modal depth of a formula w.r.t. $\I$ is one, that is, $I$ takes only the elements of $\LangCPL$ in its scope. It is done so to exclude expressions about higher order intentions, such as intention to intend. The reason for that restriction is to be discussed later. $\Nec \varphi$ read as ``\textit{it is necessarily true that $\varphi$}''. We employ the usual abbreviations for propositional connectives $\imp$ and $\biimp$ as  $\varphi\imp \psi:= \neg \varphi\vee \psi$,  and $\varphi\biimp \psi:= (\varphi\imp \psi) \wedge (\psi\imp \varphi)$.   We interpret this language on the so-called {\em problem-sensitive Kripke models}, obtained by endowing Kripke models with a join-semilattice that represents decision problems and their parthood relation. The former is familiar in modal logic and requires no elaboration \cite{blackburn2006handbook}. The component for the representation of decision problems, called {\em problems model}, is a variation of topic-sensitive models \cite{berto2022topics} and warrants further explanation. 

%aybuke{You could just copy the definition in Sect 4.2. I think $s$ there is a bit more tidily written.}
\begin{definition}[Problems model]
A problems model $\mathcal{P}$ is a tuple $(P, \oplus, s)$, where:
\begin{enumerate}
    \item $P \neq \emptyset$ is a non-empty set of decision problems;
    \item $\oplus:P \times P \rightarrow P$ is a binary idempotent, commutative, associative operation: problem fusion. We assume the unrestricted fusion, that is, $\oplus$ is always defined on $P$: $\forall A \subseteq P \exists a \in P: \bigoplus A = a$. The parthood relation $\leq$ is defined in terms of $\oplus$: $\forall a,b \in P: a \leq b $ iff $a \oplus b = b$;
    \item  $s : \Prop\rightarrow 2^P$ is a function assigning a set of decision problems to each element in $\Prop$. Namely, if $a \in s(p)$, then bringing it about that $p$ is a suitable (partial) solution to $a$. $s$ extends to the whole propositional  language $\LangCPL$ as follows:
    
    \begin{center}
\begin{tabular}{lll} 
$\s(\top) = P$ & \,\,\,\, & $\s(\varphi \lor \psi) =  \s(\varphi) \cap \s(\psi)$\\
$\s(\neg \varphi) = \s(\varphi)$ &\,\,\,\, & $\s(\varphi \land \psi) = \{ a \oplus b \:| \: a \in \s(\varphi), b \in \s(\psi) \}$
\end{tabular}
\end{center}

\begin{comment}
    \begin{align*}
    &\s(\top) = P\\
    &\s(\neg \varphi) = \s(\varphi)\\
    &\s(\varphi \lor \psi) =  \s(\varphi) \cap \s(\psi)\\
    &\s(\varphi \land \psi) = \{ a \oplus b \:| \: a \in \s(\varphi), b \in \s(\psi) \}.
    \end{align*}
\end{comment}
Given any $p \in Prop$, $s(p)$ is upward closed:  $a \leq b \Rightarrow (a \in \s(p) \Rightarrow b \in \s(p))$.
\end{enumerate}
\label{def:atomic_problems_revised}    
\end{definition}

The extension of $s$ to the whole propositional language is motivated as follows. Given any problem $a \in P$, forcing $\varphi$ is a partial solution to the problem iff to decide if $\varphi$ or not to $\varphi$ is a part of $a$, hence, $s(\varphi) = s(\neg \varphi)$. $\varphi \lor \psi$ partially solves $a$ iff both disjuncts -- to bring about that $\varphi$, to bring  about that $\psi$ -- solve $a$:  this corresponds to our intuition about partial solutions as indeterministic choice between a set of solutions. Finally, $\varphi \land \psi$ is a partial solution to $a$ iff $a$ can be split into two problems, i.e. $a = a_1 \oplus a_2$, such that $\varphi$ing partially solves $a_1$ and $\psi$ing partially solves $a_2$.  The following lemma shows that $s$ is upward closed over the whole language $\LangCPL$.

\begin{lemma}\label{lemma:MON}
Given a problems model $\cP=(P, \oplus, s)$, $a, b\in P$ and $\varphi\in \LangCPL$, the following holds: 
 \begin{equation}\label{eqn.mon.whole}
        a \leq b \Rightarrow (a \in \s(\varphi) \Rightarrow b \in \s(\varphi)). \notag
\end{equation}
\end{lemma}
\begin{proof}
The proof follows by induction on the structure of $\varphi$. The case for atomic propositions holds by the definition of $\s$. Case $\varphi:=\top$ is trivial as $\s(\top)=P$ and case $\varphi:=\neg\psi$ follows from the induction hypothesis (IH) and the fact that $\s(\neg\psi)=\s(\psi)$. 

\smallskip

\noindent Case $\varphi:=\psi\lor \chi$: Suppose that $a\leq b$ and $a\in \s(\psi\lor\chi)$. By the defn. of  $\s$, this means that $a\in \s(\psi)$ and $a\in \s(\chi)$. Then, by the IH, we have $b\in \s(\psi)$ and $b\in \s(\chi)$, thus, $b\in \s(\psi)\cap \s(\chi)=\s(\psi\lor \chi)$.

\smallskip

\noindent Case $\varphi:=\psi\land \chi$: Suppose that $a\leq b$ and $a\in \s(\psi\land\chi)$. By the defn. of  $\s$, this means that $a=d\oplus c$ for some $d\in \s(\psi)$ and $c\in \s(\chi)$. Observe that $a=d\oplus c$ implies that $d\leq a$ and $c\leq a$, and, in turn, $d\leq b$ and $c\leq b$. Then, by IH, we have $b\in \s(\psi)$ and $b\in \s(\chi)$. As $b=b\oplus b$, by defn. of $\s$, we have $b\in \s(\psi\land\chi)$.
\end{proof}

%aybuke{Why is this proposition important? Elaborate. If it is only to prove another result, put it as a lemma right before that result.}

\subsection{Semantics, Axiomatization, and Completeness}

\begin{definition}[Problem-sensitive frames and models]

A problem-sensitive frame is a tuple $F = (W, R, \mathcal{P}, f)$, where:
\begin{itemize}
    \item $W \neq \emptyset$ is a non-empty set of possible worlds;
    \item $R \subseteq W \times W$ is a serial binary relation on $W$, such that $R(w)$ is the set of conative alternatives to $w$, i.e. possible worlds where agent's intention in $w$ is satisfied;
    \item $\mathcal{P} = (P, \oplus, s)$ is a problems model;
    \item $f: W \rightarrow P$ is a function, which takes a world and returns a decision problem that the agent is solving in the given world;
\end{itemize}

As usual, any frame $F$ may be extended to a model $\cM = (F, V)$, where $V: \Prop\rightarrow 2^W$ is a standardly defined valuation function. %, ascribing every propositional variable a set of possible worlds where the corresponding proposition is true.
\end{definition}

%\daniil{I am not sure if we should keep the result that I am talking about in the following paragraph. It is interesting, but not essential to what we present here, and we need to mind page count.}
We show the connection between the decision problems as atomic objects and as partitions. Namely,  every atomic decision problem $a \in P$, as it is given in Definition \ref{def:atomic_problems_revised}, can be transformed into a partition $\Pi_a$. Intuitively, $\Pi_a$ is  an \textit{extensional characterization} of $a$: if we accepted the extensional parthood claim, then $\Pi_a$ could be identified with $a$ and treated as in Definition \ref{def:definability_revised}.  Moreover,  we can show that if $a \leq b$, then $\Pi_a \sqsubseteq \Pi_b$. In other words, $\leq$ relation is not arbitrary, but meaningful w.r.t. solutions:  if $a \leq b$, then every partial solution to $a$ is a partial solution to $b$, hence, $a \sqsubseteq b$. Nevertheless, the converse does not hold: we may have decision problems $a,b \in P$, such that $\Pi_a \sqsubseteq \Pi_b$, but $a \not\leq b$. This corresponds to our denial of the extensional parthood claim: extensional parthood does not entail parthood in general. See Appendix \ref{part:s4.2} for the constructions and proofs of the mentioned properties. We interpret $\LangInt$ on  problem-sensitive models as in the following definition.

\begin{definition}[Semantic for $\LangInt$]

Given a problem-sensitive model $\cM=(W, R, \mathcal{P}, f, V)$ and a possible world $w\in W$, the satisfiability relation $\models$ is defined as follows, where $\br{\alpha}_{\cM}=\{w\in W \ | \ \cM, w\models \varphi\}$:
\[
\begin{array}{llll}
\cM, w \models p  & \mbox{ iff } & w \in V(p)\\
    \cM, w \models \neg \varphi & \mbox{ iff } & \cM, w \not\models \varphi\\
    \cM, w \models \varphi \lor \psi & \mbox{ iff } & \cM, w \models \varphi \mbox{ or } \cM, w \models \psi\\
   \cM, w \models \Nec \varphi & \mbox{ iff } & \forall v \in W: \cM, v \models \varphi\\
    \cM, w \models \I \alpha & \mbox{ iff } & R(w) \subseteq \br{\alpha}_{\cM} \mbox{ and } f(w) \in \s(\alpha).\\
\end{array} \]
\label{def:sat_for_problem_sensitive_models}
\end{definition}

\noindent The notions of logical consequence, validity, soundness, and completeness are defined standardly \cite{Blackburn_Rijke_Venema_2001}. While the semantic clauses for the Booleans and $\Nec$ are standard, the intention operator $\I\varphi$ is interpreted in a problem-sensitive way: the agent intends that $\alpha$ at $w$ iff (1) $\alpha$ is true in all conative alternatives at $w$ and (2) $\alpha$ is a partial solution to the agent's decision problem at $w$. (2) is what  invalidates \ref{CLE}-\ref{AEq}, making the resulting logic hyperintensional. 

\begin{theorem}
Principles \ref{CLE}-\ref{AEq} are invalid (also for atomic propositions). 
\end{theorem}
\begin{proof}
We provide a counterexample against \ref{AEq}. The same counterexample can be used to prove the invalidity of the others. Consider the problem-sensitive model $\cM = (\{w\}, \{(w, w)\}, (\{a, b, c\}, \oplus, \s), f, V)$ such that $a\oplus b=c$, $a\not \leq b$ and $b\not \leq a$, $\s(p)=\{a, c\}$, $\s(q)=\{b, c\}$, $f(w)=a$,  and $V(p)=V(q)=\{w\}$. We then have that $\cM\models \Nec(p\leftrightarrow q) \wedge \I p$, but $\cM, w\not \models \I q$, since $f(w)=a\not \in \s(q)$.
\end{proof}

%When the model is clear, we skip the subscript and write only $\br{\alpha}$. Given any problem-sensitive model $\cM = \Model$, $\cM \models \varphi$ means that $\cM,w \models \varphi$ for any $w \in W$; By $F, w \models \varphi$ we mean that $(F,V,g ), w\models \varphi$ for any valuation functions $V, g$;  By $F \models \varphi$ we mean that $F, w \models \varphi$ for any $w \in Dom(F)$; $\mathbf{Ext} \models \varphi$ means that $M \models \varphi$ for any $M \in \mathbf{Ext}$.

To state some of the axioms and other relevant principles in the complete system,  we use the abbreviation `$\overline{\varphi}$' to denote the tautology $\bigwedge_{p\in Var(\varphi)}(p\vee \neg p)$\footnote{In order to have a unique definition of each $\overline{\varphi}$, we set the convention that elements of $Var(\varphi)$ occur in $\bigwedge_{p\in Var(\varphi)}(p\vee \neg p)$ from left-to-right in the order they are enumerated in $\Prop=\{p_1, p_2, \dots\}$.  % For example, for $\varphi:=\Nec(p_3\rightarrow p_2)\vee \I^{p_1} p_5$,  $\bar{\varphi}$ is  $(p_1\vee \neg p_1)\wedge (p_2\vee \neg p_2)\wedge (p_3\vee \neg p_3)\wedge (p_5\vee \neg p_5)$, and not $(p_5\vee \neg p_5)\wedge(p_1\vee \neg p_1)\wedge  (p_3\vee \neg p_3)\wedge (p_2\vee \neg p_2)$ or $(p_3\vee \neg p_3)\wedge (p_5\vee \neg p_5)\wedge (p_2\vee \neg p_2)\wedge (p_1\vee \neg p_1)$ etc. This convention will eventually not matter since our logics cannot differentiate two conjunctions of different order: $\varphi\wedge \psi$ provably and semantically equivalent to $\psi\wedge \varphi$.
},  following a similar idea in \cite{Giordani18}, where $Var(\varphi)$ is the set of atomic propositions occuring in $\varphi$.  It is easy to see that $\I\overline{\varphi}$ simply expresses that {\em $\varphi$ is a partial solution to the agent's decision problem} ($\cM, w\models \I\overline{\varphi}$ 
 iff $f(w)\in \s(\varphi)$). We list the axioms and  rules of the logic of problem-sensitive models in Table \ref{tbl:ax:plainbelief2}. 
\begin{comment}
\begin{table}[h!]
\begin{center}
\begin{tabularx}{\textwidth}{>{\hsize=0.1\hsize}X>{\hsize=1.7\hsize}X>{\hsize=0\hsize}X}
\toprule
(CPL) & all classical propositional tautologies and Modus Ponens & \\
($\mathsf{S5}_\Nec$) & $\mathsf{S5}$ axioms and rules for $\Nec$ & \\
($\mathsf{Ax1}$) & $\I \top$ & \\
($\mathsf{Ax2}$) & $\I \varphi \imp \I{\overline{\varphi}}$\\
($\mathsf{Ax3}$) & $\I \varphi \rightarrow \neg I \neg \varphi$ & \\
($\mathsf{Ax4}$)  & $(\I \varphi \land \I \psi) \leftrightarrow \I (\varphi \land \psi)$ & \\
($\mathsf{Ax5}$) & $\Nec(\psi \rightarrow \varphi)  \rightarrow ((\I \psi \land \I\overline{\varphi}) \rightarrow \I \varphi)$ & \\
       \hline
\end{tabularx}
\end{center}
\caption{Axiomatization $\Log$ for the logic of problem-sensitive models.}\label{tbl:ax:plainbelief} \label{tbl:ax:plainbelief2}
\end{table}%
\end{comment}
\begin{table}[h!]
\centering
\begin{tabularx}{\textwidth}{lX lX}
\toprule
\textbf{Label} & \textbf{Axiom / Rule} & \textbf{Label} & \textbf{Axiom / Rule} \\
\midrule
(CPL) & All classical propositional taut. and MP 
      & ($\mathsf{S5}_\Nec$) & $\mathsf{S5}$ axioms and rules for $\Nec$ \\
($\mathsf{Ax1}$) & $\I \top$ 
      & ($\mathsf{Ax4}$) & $(\I \varphi \land \I \psi) \leftrightarrow \I (\varphi \land \psi)$  \\
 ($\mathsf{Ax2}$) & $\I \varphi \rightarrow \I{\overline{\varphi}}$ 
      & ($\mathsf{Ax5}$) & $\Nec(\psi \rightarrow \varphi) \rightarrow ((\I \psi \land \I \overline{\varphi}) \rightarrow \I \varphi)$   \\
 ($\mathsf{Ax3}$) & $\I \varphi \rightarrow \neg \I \neg \varphi$  
      & & \\
\bottomrule
\end{tabularx}
\caption{Axiomatization $\Log$ for the logic of problem-sensitive models.}\label{tbl:ax:plainbelief} \label{tbl:ax:plainbelief2}
\end{table}
We briefly comment on the axioms for the intention operator. $\mathsf{Ax1}$ means that one always has at least the weakest intention possible: to bring it about that $\top$. $\mathsf{Ax2}$ states that one intends that $\varphi$ only if $\varphi$ is a partial solution to one's decision problem, reflecting problem-sensitivity of intention. $\mathsf{Ax3}$ says that intentions are consistent and $\mathsf{Ax4}$ states that intentions agglomerate. Finally $\mathsf{Ax5}$ is our restricted closure principle: one intends those a priori consequences of one's intensions which also constitute partial solutions to one's decision problem. We conclude the section by stating our main technical result, a detailed proof of which can be found in Appendix \ref{app:1}.\footnote{Unsurprisingly, the axiomatization is similar to the axiomatization of the logic of simple hyperintensional belief in \cite{ozgun2021dynamic}. We do not have an intention counterpart of their axiom $B\varphi\rightarrow \Nec B\varphi$, since the semantics of $I$ depends on the actual world.}

\begin{theorem}\label{thm:comp}
 $\Log$ is a sound and strongly complete axiomatization of $\LangInt$ with respect to the class of all problem-sensitive models.
\end{theorem}

\section{Conclusions and Future Work}

In this paper, we developed a hyperintensional logic of intention, addressing the problem of side-effects by invalidating a number of closure principles. By integrating elements from inquisitive semantics and topic-sensitive epistemic frameworks, we provided a formal system that better captures the rational commitments of intentional agents. Our approach refines question-sensitive theory of intention by capturing hyperintensional differences between decision problems. Technically, we introduced a bi-modal logic with a sound and strongly complete axiomatization.

Several directions remain open for future work. First, incorporating belief alongside intention would allow for a more nuanced interaction between epistemic states and decision-making. Second, formalizing a notion of control would further clarify what propositions agents in general can intend. Third, allowing nested $\I$ operators would raise interesting questions, such as which decision problems could be solved by adopting  intention to $\varphi$, and how they are connected with problems solvable by $\varphi$ itself. Moreover, extending our system to dynamic settings could provide insights into intention revision and deliberative reasoning over time. Finally, a multi-agent extension of the proposed framework would allow us to model interesting interactions among intending agents, e.g.,  lack of control that results from the actions of others.

\section*{Acknowledgments}
Many thanks to our anonymous referees for their constructive feedback, which helped us significantly improve the final version of this extended abstract. Special thanks to Franz Berto for allowing us to test some of our examples on him.
\begin{comment}

\begin{itemize}
\item Summarize the conceptual and technical contributions
\item Link to other work: inquisitive semantics, question sensitive treatment of knowledge and belief, topic-sensitive treatment of knowledge and belief 
\item Future work: how to add belief, how to formalize a notion of control?

\end{itemize}
\end{comment}
%nocite{*}
\bibliographystyle{eptcs}
\bibliography{lit}
%printbibliography
\appendix
\section{Proofs from Section 4.2}\label{part:s4.2}

\begin{lemma}
    Given a problems model $\cP=(P, \oplus, s)$, a decision problem $a\in P$, and formulas $\varphi, \psi\in \LangCPL$, $a \in s(\varphi \lor \psi) \mbox{ iff } a \in s (\varphi \land \psi)$.
\label{prop:ext_part_conj_disj}
\end{lemma}
\begin{proof}

\noindent ($\Rightarrow$) Assume $a \in s(\varphi \lor \psi)$. By Defn. \ref{def:atomic_problems_revised}, this means that $a \in s(\varphi) \cap s(\psi)$, i.e. $a \in s(\varphi)$ and $a\in s(\psi)$. Then, since $a=a\oplus a$, we obtain that $a \in s(\varphi \land \psi)$.

($\Leftarrow$) Assume $a \in s(\varphi \land \psi)$. By Defn. \ref{def:atomic_problems_revised}, this means that there is $b, c \in P$ such that $a = b \oplus c, b \in s(\varphi), c \in s(\psi)$. The fact that $a = b \oplus c$ implies that $b \leq a$ and $c \leq a$. Then, by Lemma \ref{lemma:MON}, we obtain that $a \in s(\varphi)$ and $a \in s(\psi)$, that is, $a \in s(\varphi) \cap t(\psi)$. By Defn. \ref{def:atomic_problems_revised}, we conclude $a \in s(\varphi \lor \psi)$.
\end{proof}

\begin{proposition}[Every decision problem yields a partition]\label{prop:app:partition}
Given any problem-sensitive model $\cM = (W, R, \mathcal{P}, f, V)$ any decision problem $a \in P$, and $\varphi\in \cL_{CPL}$, there exists a partition $\Pi_a$ of $W$ such that if $a \in s(\varphi)$, then $\varphi$ is defined on $\Pi_a$, i.e.  $sm(\varphi) \sqsubseteq \Pi_a$.
\label{prop:atomic_problems_yield_partitions}
\end{proposition}
\begin{proof}
    Let $s^{-1}(a) = \{ \varphi \in \mathcal{L}_{CPL} \: |\: a \in s(\varphi)\}$ that is ,  $\varphi \in s^{-1}(a)$ iff bringing it about that $\varphi$ partially solves $a$. Let $S^{-1}(a)$ be a collection of maximally satisfiable subsets of $s^{-1}(a)$, i.e. $\Gamma \in S^{-1}(a)$ iff 1) $\Gamma \subseteq s^{-1}(a)$; 2) there exists a world $w \in W$, such that $\cM, w \models \varphi$ for all $\varphi \in \Gamma$ (from now on, abbreviated as $\cM, w \models \Gamma$) and 3) for any $\Delta \subseteq s^{-1}(a)$, such that $ \Gamma \subset \Delta$, there is no world $w \in W$, such that $\cM, w \models \Delta$. 

    For every $\Gamma \in S^{-1}(a)$ we define a set $\llbracket \Gamma \rrbracket = \{w \in W \: | \: \cM, w \models \varphi \mbox{ for all } \varphi\in \Gamma\}$. Then, clearly the set $\Pi_a= \{ \llbracket \Gamma \rrbracket \: |\: \Gamma \in S^{-1}(a)\}$ is a partition of $W$. We show that for any $\varphi \in \mathcal{L}_{CPL}$, if $a \in s(\varphi)$ then $\varphi$ is defined on $\Pi_a$. 

    By induction on $\varphi$. Let $\varphi := p$.  Assume $a \in s(p)$ for some proposition $p \in \Prop$. W.l.o.g., let $\llbracket p \rrbracket \neq W$.   By Definition \ref{def:subject_matters}, $sm(p) = \{ \llbracket p \rrbracket, W \setminus \llbracket p \rrbracket\}$.  Let $\widehat{p} = \{\bigcup \llbracket \Gamma \rrbracket \: |\: \Gamma \in S^{-1}(a), p \in \Gamma\}$. Clearly, $\widehat{p} = \llbracket p \rrbracket$: since $p \in s^{-1}(a)$, for every $\Gamma \in S^{-1}(a)$ either $p \in \Gamma$ or $\neg p \in \Gamma$. Hence, for any possible world $w \in W$, if $\cM, w \models p$, then there exists $\Gamma \in S^{-1}(a)$, such that $\cM, w \models \Gamma$, from which it follows that $w \in \widehat{p}$. On the other hand, if $w \in \widehat{p}$, then there exists a maximally satisfiable set $\Gamma \in S^{-1}(a)$, such that $p \in \Gamma$. Since $\cM, w \models \Gamma$, $\cM, w \models p$, from which, by definition, it follows that $w \in \llbracket p \rrbracket$.  Then, $\llbracket p  \rrbracket = \bigcup_{i \in I} X_i$ for some $\{X_i \}_{i \in I} \subseteq \Pi_a$, and $W \setminus \llbracket p \rrbracket = \bigcup_{j \in J} X_j$, where $ \{X_j\}_{j \in J} = \Pi_a \setminus \{X_i\}_{i \in I}$ i.e. $sm(p) \sqsubseteq \Pi_a$, i.e. $p$ is defined on $\Pi_a$.

    Let $\varphi : = \neg \psi$. By IH, if $a \in sm(\psi)$, then $\psi$ is defined on $\Pi_a$. Since $sm(\psi) = sm(\neg \psi) = sm(\varphi)$, if $a \in sm(\varphi)$, then $sm(\varphi) \sqsubseteq  \Pi_a$.

    Let $\varphi : = \psi_1 \lor \psi_2$ and $a \in s(\psi_1 \lor \psi_2)$. Since $s(\psi_1 \lor \psi_2) = s(\psi_1) \cap s(\psi_2)$, $a \in s(\psi_1)$ and $a \in s(\psi_2)$. By IH, if $a \in s(\psi_1), s(\psi_2)$, then $\psi_1$ and $\psi_2$ are defined on $\Pi_a$, i.e. $sm(\psi_1) \sqsubseteq \Pi_a$ and $sm(\psi_2) \sqsubseteq \Pi_a$. By Definition \ref{def:parthood_on_sm}, it means that  $\forall A \in sm(\psi_1) \exists \{ X_i\}_{i \in I} \subseteq \Pi_a: \bigcup_{i \in I} X_i = A$  and $\forall B \in sm(\psi_1) \exists \{ X_j\}_{j \in J} \subseteq \Pi_a: \bigcup_{i \in J} X_j = B$. By Definition \ref{def:subject_matters}, $sm(\psi_1 \lor \psi_2) = \{A \cap B \: |\: A \in sm(\psi_1), B \in sm(\psi_2)\} \setminus \{\emptyset\}$.  Take any $A \cap B \in sm(\psi_1 \lor \psi_2)$. As we have shown before, $A = \bigcup_{i \in I} X_i$ and $B = \bigcup_{j \in J} X_j$ for some $\{X_k\}_{k \in I \cup J} \subseteq \Pi_a$. Hence, given that all $X_i$'s and $X_j$'s are disjoint\footnote{For the contradiction, assume they are not disjoint. Then, there exists $X_i = \llbracket \Gamma_i \rrbracket$, $X_j = \llbracket \Gamma_j \rrbracket$, such that $X_i \cap X_i \neq \emptyset$. Let $w \in X_i \cap X_j$. Then, $\cM, w \models \Gamma_i \cup \Gamma_j$, which contradicts the fact that $\Gamma_i$ and $\Gamma_j$ are \textit{maximally satisfiable} subsets of $s^{-1}(a)$.}, $A \cap B = \bigcup_{l \in I \cap J} X_l$. Since $A \cap B$ was an arbitrary element of $sm(\psi_1 \lor \psi_2)$, we can conclude that $\forall O \in sm(\psi_1 \lor \psi_2) \exists \{X_l\}_{l \in L} \subseteq \Pi_a: \bigcup_{l \in L} X_l = O$, so that, $sm(\psi_1 \lor \psi_2) \sqsubseteq \Pi_a$, which means that $\psi_1 \lor \psi_2$ is define don $\Pi_a$.

    Let $\varphi := \psi_1 \land \psi_2$. Analogously to the previous case (cases are analogous in light of Lemma \ref{prop:ext_part_conj_disj}).
\end{proof}
\begin{proposition}
    Given any problem-sensitive model $\cM = (W, R, \mathcal{P}, f, V)$ and any two problems $a, b\in P$, if $a \leq b$, then $\Pi_a \sqsubseteq \Pi_b$, where $\Pi_a$ and $\Pi_b$ are defined as in Proposition \ref{prop:app:partition}.
\label{prop:ext_parthood_is_preserved}
\end{proposition}
\begin{proof}
Let $a \leq b$ for some arbitrary decision problems $a, b \in P$. Then, by Lemma \ref{lemma:MON}, if $a \in s(\varphi)$, then $b \in s(\varphi)$ for any $\varphi \in \LangCPL$. In other words, $s^{-1}(a) \subseteq s^{-1}(b)$. 

Let $X \in \Pi_a$ be an arbitrary complete solution to $\Pi_a$. By the definition of $\Pi_a$, $X = \llbracket\Gamma\rrbracket$, where $\Gamma$ is a maximally satisfiable subset of $s^{-1}(a)$. Since $s^{-1}(a) \subseteq s^{-1}(b)$, there exists some maximally satisfiable subsets $\Sigma \subseteq s^{-1}(b)$ such that $\Gamma \subseteq \Sigma$. Let $[\Gamma]_b = \{ \Sigma \in S^{-1}(b) \: |\: \Gamma \subseteq \Sigma  \}$ be a collection of such maximally satisfiable subsets. As we have just shown, $[\Gamma]_b \neq \emptyset$. Moreover, as we are to show,  $\bigcup_{\Sigma \in [\Gamma]_b} \llbracket \Sigma \rrbracket = \llbracket \Gamma \rrbracket$. From left-to-right: assume $w \in \bigcup_{\Sigma \in [\Gamma]_b} \llbracket \Sigma \rrbracket$. Then, $w \in \llbracket \Sigma_i \rrbracket$ for some $\Sigma_i \in [\Gamma_b]$. Then, $\cM, w \models \Sigma_i$. Since $ \Gamma \subseteq \Sigma_i$, $\cM, w \models \Gamma$, i.e. $w \in \llbracket \Gamma \rrbracket$. Right-to-left: assume $w \in \llbracket \Gamma \rrbracket$. Then, there exists some $\Sigma_i \in [\Gamma_b]$, such that $\cM, w \models \Sigma_i$.\footnote{In order to see why it holds, note that for any problem $x \in P$, $s^{-1}(x)$ is closed under negation, i.e., for any formula $\varphi \in \LangCPL$,  $\varphi \in s^{-1}(x)$ iff $\neg \varphi \in s^{-1}(x)$ (since $s(\varphi) = s(\neg \varphi)$). Given that and the fact that $s^{-1}(a) \subseteq s^{-1}(b)$, we can complete any maximally satisfiable set $\Gamma \in S^{-1}(a)$ to some maximally satisfiable set $\Sigma \in S^{-1}(b)$ by an analogue of Lindenbaum construction.} Then, $w \in \llbracket \Sigma_i \rrbracket$ and since $\llbracket \Sigma_i \rrbracket \subseteq \bigcup_{\Sigma \in [\Gamma]_b} \llbracket \Sigma \rrbracket$, $w \in \bigcup_{\Sigma \in [\Gamma]_b} \llbracket \Sigma \rrbracket$.

Moreover, by definition of $\Pi_b$, $\{ \llbracket \Sigma \rrbracket \: |\: \Sigma \in [\Gamma]_b \} \subseteq \Pi_b$. Hence, $X$, a complete solution to $\Pi_a$, is a partial solution to $\Pi_b$. Since $X$ was arbitrary, we can conclude that any complete solution to $\Pi_a$ is a partial solution to $\Pi_b$, from which, by Definition \ref{def:parthood_on_sm}, it follows that $\Pi_a \sqsubseteq \Pi_b$. 
\end{proof}

\begin{proposition}[$\Pi_a \sqsubseteq \Pi_b$ does not entail $a \leq b$]

Given any problem-sensitive model $\cM = (W, R, \mathcal{P}, f, V)$ and two problems $a, b\in P$, the fact that $\Pi_a \sqsubseteq \Pi_b$ does not entail $a \leq b$.
\label{prop:ext_parthood_does_not_entail_parthood}
\end{proposition}

\begin{proof}
We prove this statement by providing a counterexample. Let $\cM = (W, R, \mathcal{P}, f, V)$ be a problem-sensitive model, such that for two propositions $p, q \in \Prop$, $V(p) = V(q)$. Then, let $a$ be a problem such that $a \in s(p)$  and for any formula $\psi \in \LangCPL$, if $\psi$ contains any propositions other that $p$, then $a \not\in s(\psi)$. Consequently, $a \not\in s(q)$. Let $b$ be a decision problem such that $b \in s(q)$  and for any formula $\psi \in \LangCPL$, if $\psi$ contains any propositions other that $q$, then $b \not\in s(\psi)$. Consequently, $b\not\in s(p)$.  Obviously, $\Pi_a = \Pi_b$, hence, $\Pi_a \sqsubseteq \Pi_b$. At the same time, $a \not\leq b$, since it is not the case that for any $\varphi \in \LangCPL$, if $a \in s(\varphi)$, then $b \in s(\varphi)$: $a \in s(p)$, $b\not\in s(p)$. 
\end{proof}
\begin{comment}
\section{Proof of Theorem \ref{lem:derivable}}
\begin{enumerate}
\item Follows from $\mathsf{Ax4}$
\item Follows from Theorem \ref{lem:derivable}.\ref{lem:derivable.1}.
\item \begin{tabular}{ll|l}
      1. &$\Nec(\varphi  \rightarrow (\varphi \lor \psi))  \rightarrow ((\I \varphi \land \I(\overline{\varphi \lor \psi})) \rightarrow \I (\varphi \lor \psi))$ & $\mathsf{Ax7}$\\
      2. & $ \varphi  \rightarrow (\varphi \lor \psi)$ & CPL\\
      3.&  $\Nec((\varphi \land \psi) \rightarrow (\varphi \lor \psi)) $ & 2, $\mathsf{S5}_\Nec$\\
      4.  & $(\I \varphi \land I \psi) \rightarrow \I (\varphi \land \psi) $ & $\mathsf{Ax4}$\\
      5.  & $ \I (\varphi \land \psi) \rightarrow \I (\overline{\varphi \lor \psi})$  & $\mathsf{Ax2}$ \\
      6.  &  $(\I \varphi \land \I \psi) \rightarrow \I (\overline{\varphi \lor \psi})$ &  4, 5, MP\\
      7. & $(\I \varphi \land \I \psi) \rightarrow \I \varphi$ & CPL\\
      8. & $(\I \varphi \land \I \psi) \rightarrow (\Nec(\varphi  \rightarrow (\varphi \lor \psi)) \land \I \varphi \land \I (\overline{\varphi \lor \psi}))$ & (3, 6, 8, CPL)\\
      9. & $(\I \varphi \land \I \psi) \rightarrow \I(\varphi \lor \psi)$ & 1, 8, MP\\
    \end{tabular}
    \end{enumerate}

\end{comment}
\section{Proof of Theorem \ref{thm:comp}: Soundness and Completeness of $\Log$}\label{app:1}

We prove completeness via a canonical model construction. Our proof closely resembles the completeness proof in \cite{Giordani18} except for the construction of the problem sensitive components. We follow \cite{Siemers21} for the construction of $P^c$ and $\oplus^c$, and the construction of $f^c$ is similar to the construction of awareness sets in awareness logics \cite{Fagin1987}.

Consistency, maximal consistency, and derivability for $\Log$ are defined standardly \cite{Blackburn_Rijke_Venema_2001}.

\begin{lemma}\label{lem:derivable} 
The following are derivable in $\Log$:
\begin{multicols}{2}
\begin{enumerate}
\item\label{lem:derivable.1} $\I \overline{\varphi} \leftrightarrow \I (\bigwedge_{p\in Var(\varphi)}\overline{p})$
\item\label{lem:derivable.2}  $\I\overline{\varphi}\imp \I\overline{\psi}$, if $Var(\psi)\subseteq Var(\varphi)$
\item \label{lem:derivable.3} $(\I\varphi\wedge \I\psi) \imp \I(\varphi \vee \psi)$
\item[] \
\end{enumerate}
\end{multicols}
\end{lemma}
\begin{proof} \ref{lem:derivable.1} follows from $\mathsf{Ax4}$; \ref{lem:derivable.2} follows from Theorem \ref{lem:derivable}.\ref{lem:derivable.1}; and \ref{lem:derivable.3} follows from Axioms 7, 4, and 2.
\end{proof}

\begin{lemma}\label{lem:prop.mcs:Log} For every maximally consistent set (mcs) $\Gamma$ of $\Log$ and $\varphi, \psi\in \LangInt$, the following hold:
\begin{multicols}{2}
\begin{enumerate} 
\item\label{lem:prop.mcs:L.0}  $\Gamma\vdash \varphi$ iff $\varphi \in \Gamma$,
\item\label{lem:prop.mcs:L.1} if $\varphi\in \Gamma$ and $\varphi\rightarrow \psi\in \Gamma$ then $\psi\in \Gamma$,
\item\label{lem:prop.mcs:L.2} if $\vdash \varphi$ then $\varphi\in \Gamma$,
\item\label{lem:prop.mcs:L.3} $\varphi\in \Gamma$ and $\psi\in \Gamma$ iff $\varphi\wedge\psi\in \Gamma$,
\item\label{lem:prop.mcs:L.4} $\varphi\in \Gamma$ iff $\neg \varphi\not\in \Gamma$.
\item[] \
\end{enumerate}
\end{multicols}
\end{lemma}
  
 \begin{lemma}[Lindenbaum's Lemma]\label{lem:lindenbaum}
 Every  $\Log$-consistent set can be extended to a maximally $\Log$-consistent one.
 \end{lemma}

 The proofs of  Lemmas \ref{lem:prop.mcs:Log} and \ref{lem:lindenbaum} are standard and we sometimes use them without explicit mention. We are now ready define our canonical model.
 
 Let $W^c$ be the set of all maximally $\Log$-consistent sets.  For each $\Gamma\in W^c$, define 
 \begin{align}
 \Gamma[\Nec] \ & :=\{\varphi\in \cL: \Nec\varphi\in \Gamma\}, \notag\\
\Gamma[I] \ & := \ \{\varphi \in \LangInt : I\psi\land \Nec (\psi\imp \varphi) \in \Gamma \mbox{ for some } \psi\in \LangCPL\} \notag
 \end{align}

Moreover, we define $\sim_\Nec$ and $\rightarrow^c_\varphi$ on $W^c$, respectively, as (1) $\Gamma \sim_\Nec \Delta  \mbox{ iff } \Gamma[\Nec] \subseteq \Delta$; and (2)
$\Gamma \imp^c \Delta \mbox{ iff } \Gamma[\I] \subseteq \Delta.$
Since $\Nec$ is an $\mathsf{S5}$ modality,  $\sim_\Nec$ is an equivalence relation \cite{Blackburn_Rijke_Venema_2001}. For any maximally $\Log$-consistent set $\Gamma$, we denote by $[\Gamma]_\Nec$ the equivalence class of $\Gamma$ induced by $\sim_\Nec$, i.e.,  $[\Gamma]_\Nec=\{\Delta\in W^c : \Gamma\sim_\Nec \Delta\}$.  The following lemma shows that $\rightarrow^c\subseteq \sim_\Nec$.

\begin{lemma}\label{lem:subset}
For all $\Gamma, \Delta\in W^c$, if $\Gamma \rightarrow^c \Delta$ , then $\Gamma \sim_\Nec \Delta$. 
\end{lemma}
\begin{proof}
Let $\Gamma, \Delta\in W^c$ such that $\Gamma \rightarrow^c \Delta$, i.e., that $\Gamma[\I]\subseteq \Delta$. Let $\varphi \in \Gamma[\Nec]$. This means that $\Nec\varphi\in \Gamma$. Then, by $\mathsf{S5}_\Nec$ (since $\vdash \Nec\varphi \biimp \Nec(\top \rightarrow \varphi)$) and Lemma \ref{lem:prop.mcs:Log}, we have that $\Nec(\top \imp \varphi)\in \Gamma$. Moreover, by $\mathsf{Ax1}$, we have that  $\I\top\in \Gamma$. Then, by the definition of $\Gamma[\I]$, we conclude that $\varphi\in \Gamma[\I]$. By the first assumption that  $\Gamma[\I]\subseteq \Delta$, we obtain $\varphi\in \Delta$. Therefore, $\Gamma\sim_\Nec\Delta$.
\end{proof}

\begin{lemma}\label{lem:conjunction}
 Given a mcs $\Gamma$, for all finite $\Phi \subseteq \Gamma[\I]$, we have $\bigwedge\Phi\in  \Gamma[\I]$.
 \end{lemma}
 \begin{proof}
 Let $\Phi=\{\varphi_1, \dots, \varphi_n\} \subseteq \Gamma[\I]$. This means that, for each $\varphi_j$ with $1\leq j\leq n$, there is a $\psi_j\in \LangCPL$ such that $\I\psi_j\wedge \Nec(\psi_j\imp \varphi_j)\in \Gamma$. Thus, $\bigwedge_{1\leq j\leq n} \I\psi_j \wedge \bigwedge_{1\leq j\leq n} \Nec (\psi_j\imp \varphi_j)\in \Gamma$. Then, by $\mathsf{Ax4}$, we obtain that $\I(\bigwedge_{j\leq n}\psi_j) \in \Gamma$. By $\mathsf{S5}_\Nec$, we also have $\Nec (\bigwedge_{j\leq n}\psi_j \imp \bigwedge_{j\leq n} \varphi_j) \in \Gamma$. Therefore, $\bigwedge \Phi\in \Gamma[\I]$.
 \end{proof}

\begin{lemma}\label{lem:consistent}
Given a mcs $\Gamma$ of $\Log$, $\Gamma[I]$ is consistent. 
\end{lemma}
\begin{proof}
Assume,  toward contradiction,  that $\Gamma[I]$ is not consistent, i.e., $\Gamma[I]\vdash \bot$.  This means that there is a finite subset $\Phi=\{\varphi_1, \dots, \varphi_n\}\subseteq\Gamma[I]$ such that $\vdash \bigwedge \Phi\supset \neg\varphi_j$ for some $j\leq n$. By Lemma \ref{lem:conjunction}, we have that $\bigwedge \Phi\in \Gamma[I]$, thus,  there is a $\psi\in \LangCPL$ such that  $\I\psi\in \Gamma$ and $\Nec(\psi\imp  \bigwedge \Phi)\in \Gamma$. Since $\vdash \bigwedge \Phi\imp \neg\varphi_j$, by $\mathsf{S5}_\Nec$, we also have  $\Nec(\psi\imp  \neg\varphi_j)\in \Gamma$. Hence, $\neg\varphi_j\in \Gamma[\I]$ too. As $\varphi_j\in \Gamma[\I]$,  we also have a $\psi'\in \LangCPL$ with $\I\psi'\in \Gamma$ and $\Nec(\psi'\imp\varphi_j)\in \Gamma$. From $\Nec(\psi\imp \neg \varphi_j)\in \Gamma$ and $\Nec(\psi'\imp\varphi_j)\in \Gamma$, by $\mathsf{S5}_\Nec$, we obtain that $\Nec(\psi\imp\neg \psi')\in \Gamma$. As $\I\psi'\in \Gamma$, by $\mathsf{Ax2}$ and Lemma \ref{lem:derivable}.\ref{lem:derivable.2}, $\I\overline{\neg \psi'}\in \Gamma$.  Therefore, $\I\overline{\neg \psi'}\in \Gamma$, $\Nec(\psi\supset\neg \psi')\in \Gamma$, $I\psi\in \Gamma$, by $\mathsf{Ax5}$, imply that $\I\neg\psi'\in \Gamma$, contradicting the consistency of $\Gamma$:  $\I\psi'\in \Gamma$ implies $\neg \I\neg \psi'\in \Gamma$, by $\mathsf{Ax3}$. Therefore, $\Gamma[\I]$ is consistent. 
\end{proof}

Given a mcs $\Gamma_0$ of $\Log$, the canonical model for $\Gamma_0$ is a tuple $\cM^c=\CModel$, where 
 \begin{itemize}
 \item $[\Gamma_0]_\Nec$ is as described above.
 \item $R^c = \rightarrow^c \cap ([\Gamma_0]_\Nec\times [\Gamma_0]_\Nec)$
 \item $\cD^c=(P^c, \oplus^c, \s^c)$, where
\begin{itemize}
\item $P^c=\cP(\Prop)$
 \item $\oplus^c:P^c \times P^c \rightarrow P^c$ such that for all $A, B\in \cP(\Prop)$, $A\oplus^c B=A\cup B$.
 \item $\s^c: \cL \rightarrow \cP(P^c)$ such that $\s^c(p)= \{A \in P^c \ | \ p\in A\}$ for all $p\in \Prop$. $\s$ extends to the language $\LangCPL$ as in Definition \ref{def:atomic_problems_revised}.
 \end{itemize}
 \item $f^c: [\Gamma_0]_\Nec  \rightarrow P^c$ such that $f^c(\Gamma)=\{p\in\Prop \ | \ I\overline{p}\in \Gamma\}$.
 \item $V^c:  \Prop \rightarrow \cP([\Gamma_0]_\Nec)$ such that $V^c(p)=\{\Gamma \in [\Gamma_0]_\Nec : p\in \Gamma\}$.
\end{itemize}

\begin{lemma}\label{lem:var}
For all $\varphi\in \LangCPL$, $\s^c(\varphi)=\{A\subseteq \Prop \ | \ Var(\varphi)\subseteq A\}$.
\end{lemma}  
\begin{proof}
The proof follows by induction on the structure of $\varphi$. The case of constant $\top$ is trivial: $Var(\top) = \emptyset$, hence, $Var(\top) \subseteq A$ for any $A \in P^c$, i.e. $\{A \subseteq \Prop \: |\: Var(\top) \subseteq A\} = \cP(\Prop) = P^c$, which by Definition \ref{def:atomic_problems_revised} is equal to $s^c(\top)$. The case for atomic propositions follows directly by the defn. of $\s^c$. Case $\varphi:=\neg\psi$ follows from the induction hypothesis (IH) and the fact that $\s^c(\neg\psi)=\s^c(\psi)$ and  $Var(\neg\psi)=Var(\psi)$. We only present the proof for Case: $\varphi:=\psi\lor \chi$ and Case: $\varphi:=\psi\land \chi$ follows similarly:

\smallskip

\noindent Case: $\varphi:=\psi\lor \chi$: 
\begin{align}
\s^c(\psi\lor \chi) & =\s^c(\psi)\cap\s^c(\chi)  \tag{by the defn. of $\s^c$}\\
& = \{A\subseteq \Prop \ | \ Var(\psi)\subseteq A\} \cap \{A\subseteq \Prop \ | \ Var(\chi)\subseteq A\} \tag{by IH}\\
& = \{A\subseteq \Prop \ | \ Var(\psi)\cup Var(\chi)\subseteq A\} \tag{simple set theory}\\
& = \{A\subseteq \Prop \ | \ Var(\psi\lor\chi)\subseteq A\} \tag{$Var(\psi\lor\chi)=Var(\psi)\cup Var(\chi)$}
\end{align}
\begin{comment}
\noindent Case: $\varphi:=\psi\land \chi$: 
\begin{align}
\s^c(\psi\land \chi) & =\{A\cup B\subseteq \Prop \ | \ A\in \s^c(\psi), B\in \s^c(\chi)\}  \tag{by the defn. of $\s^c$}\\
& = \{A\cup B\subseteq \Prop \ | \ Var(\psi)\subseteq A, Var(\chi)\subseteq B\}  \tag{by IH}\\
& = \{A\subseteq \Prop \ | \ Var(\psi)\cup Var(\chi)\subseteq A\} \tag{simple set theory}\\
& = \{A\subseteq \Prop \ | \ Var(\psi\land\chi)\subseteq A\} \tag{$Var(\psi\land\chi)=Var(\psi)\cup Var(\chi)$}
\end{align}
\end{comment}
\end{proof}

\begin{corollary}\label{cor:mon:canonical}
For any $A, B\in \D^c$ and $\varphi\in \LangCPL$, if $A\subseteq B$ and $A\in \s^c(\varphi)$, then $B\in \s^c(\varphi)$.
\end{corollary}

\begin{lemma}\label{lem:canonicalmodel}
Given a mcs $\Gamma_0$, the canonical model $\cM^c=\CModel$ for $\Gamma_0$ is a problem-sensitive model.  
\end{lemma}
\begin{proof}
Obviously functions and operations $\s^c, \oplus^c, f^c, g^c$ are well-defined. Since $\oplus^c$ is defined as $\cup$ on $\cP(\Prop)$, it is idempotent, commutative, associative, and it satisfies unrestricted fusion. By Corollary \ref{cor:mon:canonical}, we know that $\s^c$ is upward closed. Seriality of $R^c$ follows from Lemmas \ref{lem:consistent}, \ref{lem:lindenbaum}, \ref{lem:subset}.
\begin{comment}
Regarding $g^c$: It is clear from the defn. that for  every $A \in \D$, $g^c(A)$ is a partition of $[\Gamma_0]_\Nec$. Moreover, as we are about to show, if $A \in \s^c(\varphi)$, then $\br{\varphi}_{\cM^c} = \bigcup_{i \in I} X_i$ for some $\{X_i\}_{i \in I} \subseteq g^c(A)$. If $\varphi=\top$, the result follows trivially since any partition exhausts the domain. Assume $A \in \s^c(\varphi)$ for some $\varphi$ that is not $\top$. Then, by Lemma \ref{lem:var}, $Var(\varphi) \subseteq A$. Then, given that $\varphi\in \LangCPL$, $\varphi$ can be rewritten as some Boolean combination of the elements in $Var(\varphi)$. By the defn. of $g^c(A)$, given that $Var(\varphi) \subseteq A$, there exists $\{X_i\}_{i \in I} \subseteq g^c(A)$, such that $\bigcup_{i \in I} X_i = \widehat{\varphi}$. Then, $\br{\varphi} = \bigcup_{i \in I} X_i $ for some $\{X_i\}_{i \in I} \subseteq g^c(A)$, since clearly $\widehat{\varphi} = \br{\varphi}$ (this follows from the truth lemma, to be shown below). Moreover, by definition of $g^c$ and $\D^c$ it is clear that if $A \subseteq B \subseteq Prop$, then $g^c(A) \sqsubseteq g^c(B)$ (larger set of atomic propositions will create a finer partition).
\end{comment}
\end{proof}

\begin{lemma}[Existence lemma]
Given a mcs $\Gamma_0$, the canonical model $\cM^c=\CModel$ for $\Gamma_0$, a world $\Delta \in [\Gamma_0]_\Nec$ and  a formula $\varphi\in \LangCPL$, if $\neg \I \varphi \in \Delta$, then either $\I \overline{\varphi} \not\in \Delta$ or there exists  a world $\Gamma \in [\Gamma_0]_\Nec$, such that $\Delta R^c \Gamma$ and $\neg \varphi \in \Gamma$.
\label{lemm:exists_lemma}
\end{lemma}
\begin{proof}
Suppose $\neg \I \varphi, \I \overline{\varphi} \in \Delta$. Since $\Delta$ is maximal and consistent, $\I \varphi \not\in \Delta$. Then, consider the set $\Delta' = \{ \neg \varphi\} \cup \Delta[\I]$. This set is consistent. To show this claim, assume otherwise. Then, $ \Delta[\I] \vdash \varphi$. In other words, there exists a finitely many formulae $\psi_1, \ldots, \psi_n\in\Delta[\I]$ such that $\psi_1, \ldots, \psi_n\vdash \varphi$.  This implies that $\Nec(\bigwedge_{1\leq i \leq n} \psi_i\imp \varphi)\in \Delta$. Moreover, by Lemma \ref{lem:conjunction}, we also have that $\bigwedge_{1\leq i \leq n} \psi_i\in \Delta[I]$. Thus, there is $\chi\in\LangCPL$ such that $\I\chi\wedge \Nec(\chi\imp \bigwedge_{1\leq i \leq n} \psi_i)\in \Delta$. Then, by $\mathsf{S5}_\Nec$, we have  $\I\chi\wedge \Nec(\chi \imp \varphi)\in \Delta$. This, together with $\I\overline{\varphi}$ and $\mathsf{Ax5}$, implies that $\I\varphi\in \Delta$, contradicting consistency of $\Delta$. Therefore, $\Delta' = \{ \neg \varphi\} \cup \Delta[\I]$ is consistent. By Lemma \ref{lem:lindenbaum}, it can be extended to a msc $\Gamma$. Obviously $\neg \varphi\in \Gamma$. Since $\Delta[\I] \subseteq \Gamma$ and $\Gamma \in [\Gamma_0]_\Nec$ (by Lemma \ref{lem:subset}), we also have $\Delta R^c \Gamma$. 
\end{proof}

\begin{lemma}\label{lem:func}
Given a mcs $\Gamma_0$, the canonical model $\cM^c=\CModel$ for $\Gamma_0$, a world $\Delta \in [\Gamma_0]_\Nec$ and  a formula $\varphi\in \LangCPL$, $\I\overline{\varphi}\in \Delta$ iff $Var(\varphi)\subseteq f^c(\Delta)$.
\end{lemma}

\begin{proof}\label{lem:f}
$\I\overline{\varphi}\in \Delta$ iff I $(\bigwedge_{p\in Var(\varphi)}\overline{p})\in \Delta$ (by $\mathsf{Ax2}$) iff $\I\overline{p}\in  \Delta$ for all $p\in Var(\varphi)$ (by $\mathsf{Ax4}$) iff $p\in f^c(\Delta)$ for all $p\in Var(\varphi)$ (by the defn. of $f^c)$ iff $Var(\varphi) \subseteq f^c(\Delta)$.
\end{proof}

\begin{lemma}[Truth lemma]\label{lemm:truth_lemma_for_problem_sensitive_logic}
Let $\Gamma_0$ be a mcs of $\Log$ and $\cM^c=\CModel$ be the canonical model for $\Gamma_0$.  Then, for all $\varphi\in\LangInt$ and $\Delta\in [\Gamma_0]_\Nec$,  we have $\cM^c, \Delta \models \varphi \mbox{ iff } \varphi\in \Delta$.
\end{lemma}
\begin{proof}
By induction on $\varphi$. Cases of propositional atoms, Boolean connectives and $\Nec$ are proven as usual.

\smallskip

\noindent Case of $\varphi: = \I \psi$.

\smallskip

\noindent ($\Leftarrow$) Assume $\I \psi \in \Delta$. Then, by $\mathsf{Ax2}$, $\I\overline{\psi}\in \Delta$. This means, by Lemma \ref{lem:func} that  $ Var(\psi) \subseteq f^c(\Delta)$. By Lemma \ref{lem:var}, this means that $f^c(\Delta) \in \s^c(\varphi)$. 
Now let $\Gamma \in [\Gamma_0]_\Nec$ such that $\Delta R^c\Gamma$, i.e., that $\Delta[\I]\subseteq \Gamma$. Observe that $\I\psi\in \Delta$ implies that $\psi\in \Delta[\I]$ (since $\Nec(\psi\imp \psi)\in \Delta$, by $\mathsf{S5}_\Nec$). By $\Delta[\I]\subseteq \Gamma$, this implies that $\psi\in \Gamma$. Then, by IH, we obtain that $\cM^c, \Gamma\models \psi$. As $\Gamma$ has been chosen arbitrarily, we obtain that $R^c(\Delta)\subseteq \br{\psi}_{\cM^c}$. Putting everything together, we conclude $\cM^c, \Delta  \models \I\psi$.

\smallskip

\noindent ($\Rightarrow$) Assume $\cM^c, \Delta \models \I\psi$. Then, by the semantics,  (1) $f^c(\Delta) \in \s^c(\psi)$ and (2) for all $\Gamma \in [\Gamma_0]_\Nec$ such that $\Delta R^c \Gamma$, $\cM^c, \Gamma \models \psi$. By (1) and Lemma \ref{lem:var}, we know that $Var(\psi)\subseteq f^c(\Delta)$. This implies, by Lemma \ref{lem:func}, that $\I\overline\psi\in \Delta$. Now, toward contradiction, assume $\I \psi \not \in \Delta$, i.e. $\neg \I \psi \in \Delta$. By Lemma \ref{lemm:exists_lemma}, there exists a world $\Gamma \in [\Gamma_0]_\Nec$, such that $\Delta R^c \Gamma$ and $\neg \psi \in \Gamma$. By IH, we have $\cM^c, \Gamma \models \neg \psi$, implying that $\cM^c, \Delta \not\models \I \psi$, which contradicts our initial assumption. Hence, $\I \psi \in \Delta$. 
\end{proof}
\begin{theorem}[Soundness and completeness]

Given any set of formulae $\Delta \subseteq \cL$ and any formula $\varphi \in\cL$, the following holds: $ \Gamma \vdash \varphi \mbox{ iff } \Gamma \models \varphi$.
\label{thm:soundness_and_completeness}
\end{theorem}
\begin{proof}
    ($\Rightarrow$) By  a routine check of axiom validities and validity preservation under the rules of inference of $\Log$. ($\Leftarrow$) By contraposition. Assume $\Gamma \not\vdash \varphi$. Then, $\Gamma \cup \{ \neg \varphi\}$ is $\Log$-consistent. By Lemma \ref{lem:lindenbaum}, there exists a maximal consistent set $\Gamma_0 \subseteq \cL$, such that $\Gamma\cup \{ \neg \varphi\} \subseteq \Gamma_0$. Then, we can build a canonical model $\mathcal{M}^c =\CModel$. Since $\Gamma \cup \{\neg \varphi\} \subseteq \Gamma_0$, by Lemma \ref{lemm:truth_lemma_for_problem_sensitive_logic}, $\mathcal{M}^c, \Gamma_0 \models \Gamma \cup \{ \neg \varphi\}$.  By Lemma \ref{lem:canonicalmodel}, $\mathcal{M}^c$ is a problem-sensitive model. Hence, there exists a pointed problem-sensitive model, where $\Gamma$ is true and $\varphi$ is false,  i.e. $\Gamma \not\models \varphi$. 
\end{proof}
%include{appendix}
\end{document}